\documentclass{article}
\usepackage[utf8]{inputenc}
\usepackage{amsmath}
\usepackage{amsthm}

\usepackage{amssymb}
\usepackage{pgf}
\usepackage{tikz}
\usetikzlibrary{automata}
\usetikzlibrary{decorations,arrows}
\usetikzlibrary{decorations.pathmorphing}
\usetikzlibrary{positioning}

\newcommand\definesymb[1]{%
\expandafter\newcommand\csname #1#1\endcsname{{\ensuremath{\mathbb{#1}}}}%
}

\definesymb{Z}
\definesymb{N}
\definesymb{R}

\newtheorem{theorem}{Theorem}
\newtheorem*{theorem2}{Theorem}
\newtheorem{definition}{Definition}[section]
\newtheorem{example}{Example}[section]
\newtheorem{proposition}{Proposition}
\newtheorem{fact}{Fact}
\newtheorem{corollary}{Corollary}

\title{Enumeration Reducibility in Closure Spaces with Applications to
 Logic and Algebra}

\author{Emmanuel Jeandel\\
    Université de Lorraine, CNRS, Inria, LORIA, F 54000 Nancy, France
  }

\begin{document}
\maketitle
\begin{abstract}
  In many instances in first order logic or computable algebra, classical theorems
  show that many problems are undecidable for general structures, but become
  decidable if some rigidity is imposed on the structure.
  For example, the set of theorems in many finitely axiomatisable theories is
  nonrecursive, but the set of theorems for any finitely axiomatisable \emph{complete}
  theory is recursive.
  Finitely presented groups might have an nonrecursive word problem, but
  finitely presented \emph{simple} groups have a recursive word problem.

  In this article we introduce a topological framework based on closure spaces
  to show that many of these proofs can be obtained in a similar setting. We will
  show in particular that these statements can be generalized to cover arbitrary
  structures, with no finite or  recursive presentation/axiomatization.
  This generalizes in particular work by Kuznetsov and others.

  Examples from first order logic and symbolic dynamics will be discussed at length.
\end{abstract}	

\section*{Introduction}

When one deals with sufficiently complicated algebraic structures, it is
customary to see uncomputability results for many questions:
A finitely presented semigroup with an undecidable word problem was presented by
Post \cite{Post}. The same result was obtained for groups by Novikov
\cite{Novikov} and Boone \cite{BooneV, BooneVI}, and for division rings by
MacIntyre \cite{MacIntyre}.

If we add however some hypothesis on the structure, we are sometimes able to
prove that the problem becomes decidable. This is the case for 
instance of finitely presented \emph{simple} groups \cite{BooneHigman}.
More generally, Kuznetsov \cite{Kuznetsov} proved that (using the vocabulary of Maltsev \cite[Theorem
4.2.2]{Maltsev}), every simple finitely presented algebra is constructive.

There are however other structures that are not algebras, 
where the exact same situation happens.
As an example, there are unsolvable finitely axiomatisable first-order theories
(for instance Robinson's Q \cite{RobinsonQ}). By contrast, any \emph{complete}
finitely axiomatisable theory is decidable (Folklore).

A similar situation happens in multidimensional symbolic dynamics, i.e. the
study of tilings of the plane: 
Some subshifts of finite type have an uncomputable language in
dimension 2 \cite{Robinson}. However \emph{minimal} subshifts of finite type have a computable
language \cite{Hochman2,ballier2008}.

The goal of this paper is to provide an unifying framework in which all the
decidable results will be seen as an instance of the same theorem.
We will use  the vocabulary of topology rather than algebras. The main reason to do
this is that using topology is already sufficient to obtain the main
theorems, and we are able this way to obtain results on structures
which cannot be made easily into algebras (in the sense of Maltsev).

\vskip 2mm 
The main idea is to put a structure on the set of all theories in a given
language, or on the set of all groups with a given generating set, or on the set
of all subshifts over a given alphabet. This
structure will be called a \emph{quasivariety} by analogy with universal
algebra. In the case of finitely generated groups, this structure is known as
the space of marked groups \cite{Grigorchuk}.

For this structure, finitely presented groups, finitely axiomatizable theories,
subshifts of finite type will correspond to the same objects, that we call 
\emph{finitely presented} points, using the vocabulary from algebra.

In this structure, simple groups, complete theories and minimal subshifts all
correspond to the same objects, \emph{maximal} points.

Our main theorem states that finitely presented maximal points in a
given quasivariety are computable, thereby generalizing all previous theorems.

By doing this abstractly we will see that it is actually possible to drop the
hypothesis that the points are finitely presented, and obtain a more general
statement when the objects in our quasivariety are not supposed to have a finite
presentation, or even a recursive presentation. Of course in this case, the
points will be unlikely to be computable. Our main result is that, in a maximal
point, false statements are \emph{enumeration-reducible} to true statements: Given an
enumeration of all statements that are true in the structure, we are able to
deduce all statements that are not true: we obtain negative
information about the structure from positive information.
\vskip 3mm

Roughly speaking, a set $A$ is enumeration reducible \cite{FriedbergRogers,Odifreddi2} to a set $B$, if
an enumeration of $A$ can be obtained in some effective way from any enumeration of $B$.
It is not surprising that this concept has an important role here, as it already
has been applied successfully, in particular in the context of groups.
We know for example that for every enumeration degree $d$, there exists a
finitely generated group for which the word problem has enumeration
degree $d$ (the reduction is actually stronger, Dobritsa  \cite[Theorem 2.4]{Belegradek}, see also
\cite{Ziegler1}) or that
enumeration reducibility characterizes, given a group $H$, 
when a group $G$ can be embedded into a group that is finitely presented over  $H$ (C.F. Miller, see \cite[Chapter 6]{HigmanScott}).
This notion has also been used in a more general context by Belegradek
\cite{Belegradek2}, and our first easy propositions about
presentations are reflected in this article.

We note in passing that various other reductions have been used in
conjunction with algebraic objects, in particular quasi-reducibility
\cite{DFN}
and Ziegler-reducibility \cite{Ziegler2}. However many of our theorems
have converses (see in particular Theorem~\ref{thm:conv}) which suggest enumeration-reducibility is indeed the
right notion in our context.

\vskip 3mm

The article is organized as follows: We first define the main concept of a
quasivariety (and its associated  closure space) in the first section.
In the next section, we will be interested in computability properties of all
possible presentations of a point.  The remaining sections are concerned with
the main theorem, namely the concept of a maximal points, and computability
properties of maximal points, and generalizations of maximal points.
The article ends with a discussion on which other properties one might try to
capture in this framework.

There are three main examples used in this paper: first-order logic, symbolic
dynamics and finitely generated groups.
We will focus on the first two, and the results for finitely generated groups
will be given in appendix.
As our results are proven in a full generality, it is quite likely that stronger
computability statements can be proven in some particular examples. This happens
for first-order logic, where we can obtain a much stronger result for the main
theorem. For the quasivariety of subshifts, various recent results show however
that our general results are actually tight.

\subsection*{Related work}

While our approach is based on closure spaces, it is of course not the
only way to obtain an unifying framework for different structures.

One possibility would be to use universal algebra. Many of the results
presented here for finitely presented algebras can be found in the
work of Maltsev and Kuznetsov \cite{Maltsev,Kuznetsov}. While
recursive properties have been investigated from the start, the
concept of enumeration-reducibility only appear in Belegradek
\cite{Belegradek2}, and is not used in its full generality.
Universal algebra however suffer from the fact that many theories we
are interested here do not fit easily in the setting. This is the case
for first order theory, and for subshifts. Without going into much
details, one of the problem with the theory of subshifts, that we
don't solve here, is how to obtain a good notation of the restriction of a subshift
over an alphabet $A$ to a subshift over an alphabet $B \subseteq
A$. In terms of universal algebra, if a subshift is generated by three
generators $a,b,c$, it is not clear what is the object generated only
by $a$ and $b$, and in particular if it is also a subshift.

Other approaches using category theory or model theory suffer from the
same problems.

One possible solution is to use Stone Duality, and to associate to
each object (group, subshift, theory, etc.) a topological space and/or
a boolean algebra. We do now know if it is possible to use this
connection to obtain our theorems here.

\section{Definitions}

\subsection{Quasivarieties}
We will assume rudimentary notions of computability theory, in
particular the notion of a recursive set of integers, a recursive
function, and the concept of a recursively enumerable set. See
\cite{Odifreddi} for details.

Let $I$ be an infinite recursive set, that we identify with the set of integers.
In applications, $I$ will be the set of finite words over a given finite
alphabet, or the set of formulas in a given finite signature.

In this article, we will always identify a subset $X \subseteq I$ and
a point $x \in \{0,1\}^I$.

We are interested in subsets $X \subseteq I$ that can be defined by
some Horn formulas , i.e. by axioms of the type:
\[ a \in X \wedge b \in X \wedge .. \wedge c \in X \rightarrow z \in X\]

In the vocabulary of Higman, these are called \emph{identical
  implications}.

In what follows, we will be given such a collection of formulas, and we
will look at the set of all $X$ that satisfy all formulas of the collection.

\begin{definition}
	Let $S$ be a recursively enumerable set of finite sequences of elements of $I$.
	
	A word $x \in \{0,1\}^I$ satisfies $S$ if for all $(n_0,
	n_1, \dots n_k)\in S$,	
\[	x_{n_1} = 1 \wedge x_{n_2} = 1 \wedge \dots x_{n_k} = 1 \rightarrow x_{n_0} = 1 \]
Equivalently, a set $X \subseteq I$ satisfies $S$ if for all $(n_0,
n_1, \dots n_k) \in S$:
\[	n_1 \in X \wedge n_2 \in X \wedge n_k \in X \rightarrow n_0 \in X \]
The  \emph{quasivariety} $V$ defined by $S$ is the set of all words $x$
 (or all subsets $X \subseteq I$) that satisfy $S$.
 
\end{definition}
The fact that  $S$ is recursively enumerable is not
mandatory, but happens in all interesting examples.
This assumption can be dropped in almost all theorems, to obtain
relativized versions of the theorems, by replacing all statements of the
form $X \leq_e Y$ by $X \leq_e Y \oplus S$, where $\oplus$ is the disjoint union: $A \oplus B = \{ (0,x), x \in A\} \cup \{ (1,x), x \in B\}$.

Before giving more properties of quasivariety, we will give a few alternate definitions.

First it is easy to see from the definitions that $V$ can be given the
structure of a topological space, by inheriting the natural
(product/Tychonoff) topology on $\{0,1\}^I$. 
For this topology $V$ is topologically closed, and even compact.
As the set $S$ of formulas that define $V$ is recursively enumerable,
$V$ is actually effectively closed (i.e. $V$ is a  $\Pi_1^0$ class) \cite{CenzerRec,Cenzerbook}.

\begin{definition}
  A set $X \subseteq \{0,1\}^I$ is effectively closed if there exists a
recursively enumerable set ${\cal F} = \{ f^i\}_{i \in \mathbb{N}}$ of
partial finite maps (i.e. $f^i \in \{0,1\}^{F_i}$ with $F_i$ finite), s.t. 
$X$ are exactly the points of $\{0,1\}^I$ that disagree with every element of ${\cal F}$:
\[
x \in X \iff \forall i \in \mathbb{N}, \exists j \in F_i, x_j \not= f^i_j
\]
\end{definition}  

\begin{fact}
  $V$ is effectively closed.
\end{fact}
\begin{proof}
  For each $(n_0, n_1 \dots n_k) \in S$, consider the map $f$
  defined by $f_{n_1} = f_{n_2} = \dots f_{n_k} = 1$ and $f_{n_0} = 0$.
  As $S$ is recursively enumerable, the set ${\cal F}$ of all functions we
  obtain this way is recursively enumerable, and proves that $V$ is effectively closed.
\end{proof}

The quasivariety $V$ has also the structure of a complete semi-lattice, as evidenced by
the following easy facts:

\begin{fact}\label{fact:clo}
Let $V$ be a quasivariety. Then any intersection of elements of $V$ is
again in $V$.
In particular:
\begin{itemize}
	\item $V$ contains a minimal element, the intersection of all elements of $V$.
	\item $V$ contains a maximal element, the set $I$ itself.
	\item For any set $Y \subseteq I$ there exists a smallest element
	  $X$ of $V$ that contains~$Y$.
\end{itemize}	  
\end{fact}
Note that a complete semi-lattice is also a complete lattice, where we
define the meet of $X$ and $Y$  to be the smallest element of $V$ that contains $X \cup Y$.

This is a characterization in the following sense:
\begin{theorem}[Alternate definition 1]
	\label{thm:pi01}
	An effectively closed set $S \subseteq \{0,1\}^I$ is a quasivariety iff it
	contains $I$ and is closed under (finite) intersections (if $X\in S$ and
        $Y\in S$ then $X \cap Y \in S$).
\end{theorem}	
We defer the proof of this theorem to the appendix, to not deviate
from the narrative.

Note however that this characterization does not mean that we are
actually investigating complete lattices that are effectively closed when seen as subsets
of $\{0,1\}^I$.
Indeed we are not interested in the lattice $V$ itself but in the pair
$(\{0,1\}^I, V)$, i.e. in how $V$ behaves as a subset of the
surrounding set $\{0,1\}^I$.

\subsection{The closure operator}
We will now give another characterization from the point of view of deductive
systems.

\begin{definition}
  For a quasivariety $V$ over a set $I$, the closure operator ${\cal C}$ associated to $V$ is the
  map from $\{0,1\}^I$ to $V$ that sends a set $R$ to the smallest set in $V$ containing $R$.
\end{definition}
This is well defined due to  Fact~\ref{fact:clo}.
Intuitively, if we see $R$ as a set of axioms, then ${\cal C}(R)$ is the set of all
consequences of $R$.

The vocabulary ``closure operator'' comes from the fact that this indeed makes $I$ a closure space:

\begin{definition}[\cite{Closure}]
  A pair $(I,C)$, where $C$ is a map from $\{0,1\}^I$ to $\{0,1\}^I$ is a closure
  space if
  \begin{itemize}
  \item For all $R \subseteq I$, $R \subseteq C(R)$.
  \item $C$ is idempotent: For all $R \subseteq I$, $C(C(R)) = C(R)$.
  \item $C$ is monotone: For all $A,B \subseteq I$, if $A \subseteq B$
    then $ C(A) \subseteq C(B)$.
  \end{itemize}
  A closure space is a Tarski space\cite{Tarski1,Tarski2} if additionally $I$ is countable and $C$ is finitary: For all $R \subseteq I$, if $x \in C(R)$ then there exists a finite $R' \subseteq R$ s.t. $x \in C(R')$.
\end{definition}  
To be accurate Tarski \cite{Tarski1} assumes further properties from the space, in particular that the set $I$ 
  itself is finitely presented (see below for what it means). All of our
  examples satisfy this assumption, and many, but not all, of our
  theorems, have it as a hypothesis.

It is customary in logic to write $X \models y$ instead of $y \in C(X)$ and we
will use this notation in a few proofs, in particular for the quasivariety of theories.
  
    \begin{fact}
Let $V$ be a quasivariety, and ${\cal C}$ its closure operator. Then $(I,{\cal C})$ is a Tarski space.
    \end{fact}
    
This gives another definition:
\begin{theorem}[Alternate Definition 2]
  Let $(I,C)$ be a Tarski space and $V$ the image of $C$.
  $V$ is a quasivariety iff $I$ is recursive and $V$ is effectively closed.
\end{theorem}
\begin{proof}
One direction is easy. Now suppose $I$ is recursive and $V$ is
effectively closed.
First, by idempotency, $V$ is exactly the set of points $X$ s.t. $C(X) = X$.
  
By the first axiom of closure space, $C(I) \supseteq I$ and therefore $C(I) = I$
and $I \in V$.

Let $X,Y \in V$, i.e. $C(X) = X$ and $C(Y) = Y$.
Then by monotonicity, $C(X \cap Y) \subseteq C(X) = X$,
and similarly $C(X \cap Y) \subseteq Y$. We therefore deduce $C(X \cap Y) \subseteq X \cap Y$
and finally $C(X \cap Y) = X \cap Y$, therefore $X \cap Y \in V$.

The result then follows from the previous theorem.
  
\end{proof}  
Note that the map ${\cal C}$ is usually not recursive, but as we will
see in the next section, it is given by a enumeration operator.

This is the last time we mention the notion of closure space, and we
will use vocabulary relevant to algebra rather than topology in the
following. Table~\ref{table:dic} gives a correspondence between the vocabularies.
\begin{table*}
  \begin{center}
\begin{tabular}{c|c}
Quasivarieties $V$  & Tarski spaces\\
\hline
$X$ is a point of $V$ & $X$ is a (deductive) system\\
$R$ is a presentation of $X \not= I$ & $R$ is consistent \\
$X \in V$ is maximal, $X \not= I$ & $X$ is maximally consistent\\
$X$ is finitely presented & $X$ is finitely axiomatizable/$X$ is compact\\
$I$ is finitely presented & $(I,V)$ is compact\\
\end{tabular}
\end{center}
\caption{Dictionary between quasivarieties (as defined in this paper) and Tarski spaces}
\label{table:dic}
\end{table*}

\section{Examples}

Before proceeding to the main definitions and theorems, we will give a
few running examples. The first three examples are central.They are
actually examples of families of quasivarieties, as each of them is
parametrized by some finite set (a signature, or a finite alphabet).

\subsection{The quasivariety $V_{FO,\tau}$ of first-order theories}

Let $\tau$ be a finite signature.  The set of all first order theories
over the signature $\tau$ can be given a structure of a quasivariety
$V_{FO,\tau}$.  (Here we define a theory as the set of all logical consequences of some (possibly empty) set of axioms).
	
Indeed, let $I$ be the set of all formulas over the signature
$\tau$, and let $S$ be the set of all formulas $(\phi_0, \phi_1,
\dots \phi_k)$ s.t. $\phi_0$ is a consequence of $\phi_1 \dots \phi_k$.
This set $S$ is indeed recursively enumerable. This defines a
quasivariety $V_{FO,\tau}$.

Any theory is immediately a point of the quasivariety $V_{FO,\tau}$, and 
by Gödel completeness theorem for first order logic, all points of
the quasivariety $V$ are indeed theories (closed under logical
consequence): if $\phi$ is a consequence of the formulas of $X \in
V_{FO,\tau}$, then $\phi$ is a consequence of finitely many formulas of $X$,
and therefore in $X$ itself.
	 
The quasivariety $V_{FO,\tau}$ of theories contains two particular points: the point
$X$ consisting of all tautologies and the point $X = I$ of all
formulas (i.e. the inconsistent theory).

Notice that we are defining different quasivarieties $V_{FO,\tau}$ depending on the
signature $\tau$. As we discuss computability, we assume here that
$\tau$ is finite. The case where $\tau$ is countably infinite can be
handled with some care.

If $\cal F$ is a fragment of the first-order language (for example if
$\cal F$ is the set of all universal sentences), it is also possible
to restrict the quasivariety to $V_{\cal F,\tau}$ by restricting to
formulae in $\cal F$. This example may be useful later on.

\subsection{The quasivariety $V_{sym}$ of subshifts}

  Let $\Sigma$ be a finite alphabet.  
A subset $Y$ of $\Sigma^\mathbb{Z}$ is called a subshift
\cite{LindMarcus} if it is
topologically closed and invariant under translation.
$Y$ is entirely characterized by the set $X$ of all \emph{forbidden}
words, i.e. finite words that do not appear in any word of $Y$.
For example if $\Sigma = \{0,1\}$ and $Y = \{ \dots 000 \dots, \dots 00100 \dots
\}$ is the set of all words with at most one occurence of the symbol $1$, then
the set of forbidden words $X$ of $Y$ is exactly the set of all words that contains
at least two occurences of the symbol $1$.

The quasivariety of all subshifts over $\Sigma$ will not be given with subshifts
as points, but with forbidden languages as points. This is of course equivalent.

A set $X$ of words over $\Sigma$ is the forbidden language of a subshift if it
is extensible and factorial, that is:
\begin{itemize}
	\item For any letter $a$, if $w \in X$ then $aw \in X$
	\item For any letter $a$, if $w \in X$ then $wa \in X$
    \item If $wa \in X$ for all letters $a$, then $w \in X$.
    \item If $aw \in X$ for all letters $a$, then $w \in X$.
\end{itemize}
By taking all these Horn formulas as our set $S$, we see that the set of
of subshifts over $\Sigma$ is a quasivariety $V_{sym}$ defined over
the set $I = \Sigma^\star$, the set of all finite words.

The quasivariety $V_{sym}$ of subshifts contains two particular points: the point
$X = \emptyset$ (which corresponds to the subshift $Y =
\Sigma^\mathbb{Z}$) and the point $X = \Sigma^\star$ (which
corresponds to the subshift $Y = \emptyset$).

Similar definitions can be given for higher dimensional subshifts
(i.e. subsets of $\Sigma^{\mathbb{Z}^d}$), which are again
characterized by the set of finite patterns that do not appear in them.

\subsection{The quasivariety $V_{grp}$ of finitely generated groups}

Let $n$ be an integer.
The set of all groups with $n$ generators may be seen as a quasivariety.
Indeed, a group $G$ with $n$ generators can be seen (up to isomorphism) as a quotient
of the free group $\mathbb{F}_n$, or equivalently as a normal subgroup
$R$ of $\mathbb{F}_n$ (the subgroup $R$ corresponds to the word problem
of $G$, i.e. all combinations of generators of $G$ that are equal to
the identity).

Indeed, a set $X \subseteq \mathbb{F}_n$ is a normal subgroup of
$\mathbb{F}_n$ (i.e. codes a group) if:
\begin{itemize}
	\item If (nothing) then $1 \in X$.
	\item If $g \in X$ then $g^{-1} \in X$.
	\item If $g \in X, h \in X$ then $gh \in X$.
	\item For any $h$, if $g \in X$ then $hgh^{-1} \in X$.
\end{itemize}	
For example, for $\mathbb{Z}^2 = \{a,b | ab=ba\}$, we have $a \not\in
X$, $b \not\in X$ but $aba^{-1}b^{-1} \in X$, $aaba^{-1}b^{-1}a^{-1} \in
X$, and more generally, $X$ is exactly the set of words of the free
group for which the number of occurences of $a$ is equal to the number
of occurences of $a^{-1}$, and the same for $b$ and $b^{-1}$.

By taking all previous Horn formulas as our set $S$, we see that the set of
of groups with $n$ generators is a quasivariety $V_{grp}$ defined over
the set $I = \mathbb{F}_n$

The quasivariety contains two particular points: the point
$X = \{1\}$ (which corresponds to the group $G = \mathbb{F}_n$)
and the point $X = \mathbb{F}_n$ (which corresponds to the one-element
group).

This particular quasivariety is usually called the space of \emph{marked
groups}, see Grigorchuck \cite{Grigorchuk}.

\subsection{The quasivariety  $V_{\mathbb{R}}$ of closed subsets of
  $\mathbb{R}$}

As quasivarieties can be defined by closure operators, it is of course
natural that topologies can be made into quasivarieties in our sense.
The astute reader might for example see how to twist the definitions
of the quasivariety of subshifts to obtain the quasivarieties of all
closed sets of the cantor space $\{0,1\}^\mathbb{N}$.

We will look at another example, the set of all closed subsets of
$\mathbb{R}$.
Similarly to the example of subshifts, it will be more natural to see
a closed subset by the set of intervals that do not intersect it.

More precisely let $I = \{ (p,q) \in  \mathbb{Q} \times \mathbb{Q}, p
< q\}$. Elements $(p,q)$ of $I$ are to be understood as open intervals
$]p,q[$.

Let $K$ be a closed subset of $\mathbb{R}$.
Then $X$ is the set of all intervals that do not intersect $K$ if and
only if:

\begin{itemize}
  \item If $]p,q[ \in X$ and $r \in ]p,q[$, then $]p,r[ \in X$  and $]r,q[ \in X$
  \item If $]p,r[ \in X$ and $]s,q[ \in X$ and $p \leq s < r$ then $]p,q[ \in X$
\end{itemize}  
The first property states that if $D_1 \subseteq D_2$ are intervals
and $D_2$ does not intersect $K$, then $D_1$ does not intersect $K$.
The second property states that if $D_1$ and $D_2$ are two overlapping
intervals that do not intersect $K$, then  $D_1 \cup D_2$  does not intersect $K$.

By taking all previous Horn formulas as our set $S$, we see that the set of
all closed subsets of $\mathbb{R}$ is a quasivariety $V_{\mathbb{R}}$.

The quasivariety contains two particular points: the point
$X = \emptyset$ (which corresponds to the closed set $\mathbb{R}$)
and the point $X = I$ (which corresponds to the closed set $\emptyset$)

We note that this construction is quite general, and will work in many
topological spaces as long as they are computable in some sense
(TODO), by taking for $I$ a suitable countable basis.
We note however that the closed subsets of the Baire Space
$\mathbb{N}^\mathbb{N}$ cannot be
put in the natural way into a quasivariety, due to
$\mathbb{N}^\mathbb{N}$ not being relatively compact.

\section{Presentations}

\begin{definition}
  A \emph{presentation} of $X \in V$ is a set $Y$ so that $X = {\cal C}(Y)$.
  
$X$ is \emph{finitely presented} if it admits a finite presentation.
$X$ is \emph{recursively presented} if it admits a recursively enumerable presentation.
\end{definition}
The vocabulary of presentations comes from algebra. From the point of view of logic, we could see
$Y$ as an axiomatization of the theory $X$.

\begin{example}
  \begin{itemize}
    \item A finitely presented theory is usually called a finitely axiomatisable
theory: It is a theory that can be given by finitely many axioms.

A recursively presented theory is usually called a recursively axiomatisable theory.

\item A finitely presented subshift is usually called a subshift of finite
type: It is a subshift given by a finite list of forbidden words
(or forbidden patterns, in higher dimensions). As an example the set $Y$ of all
biinfinite words over the alphabet $\{0,1\}$ where every symbol $0$ is preceded by a symbol $1$ is a subshift of finite
type and given by the set of forbidden words $\{00\}$. The set $Y$ of all
biinfinite words
with at most one symbol $1$ is \emph{not} of finite type.

A recursively presented subshift is usually called an effectively
closed subshift.

\item Our vocabulary coincides with the vocabulary from group theory,
  so that finitely (resp. recursively) presented groups in the quasivariety of groups are
  exactly the finitely (resp. recursively) presented groups

\item Finitely presented closed subsets of $\mathbb{R}$ are
  complements of finite union of open (finite) intervals with rational coefficients.
  Equivalently finitely presented closed subsets of $\mathbb{R}$ are
  of the form $]-\infty, r] \cup K \cup [s,+\infty[$ where $K$ is a
  finite union of closed intervals (with rational coefficients) and
  $r,s \in \mathbb{Q}$.
  In particular the empty set is not finitely presented.
        
  Recursively presented closed subsets of
  $\mathbb{R}$ are usually called $\Pi_1^0$ subsets of $\mathbb{R}$,
  or effectively closed sets.
\end{itemize}
\end{example}
In what follows, we are interested in what computability properties on a presentation of $X$ transfer to
$X$.
It is well known for example that the set of true formulas of a finitely axiomatisable (or
recursively axiomatisable) theory is recursively enumerable, but what happens if our theory is not
finitely axiomatisable, or even worse, not recursively axiomatisable?

It turns out that the good notion to use in such a context is the notion of enumeration
reducibility.
While the formal definition is a bit cumbersome, the intuitive idea is as follows:
A set $A$ is enumeration-reducible to a set $B$ if there is some procedure that can enumerate the
elements of $A$ given \emph{any}  enumeration of the elements of $B$. By an enumeration of $B$ we
mean any possible (typically non recursive) way to give elements of $B$ one at a time.
It is not required that the elements of $A$ are enumerated in a specific order, and this order
will actually usually depend on the order in which the elements of $B$ are given.

\begin{definition}[\cite{FriedbergRogers,Odifreddi2}]
Let $A,B$ two subsets of $I$.

We say that $A$ is enumeration reducible to $B$, written ${A \leq_e B}$,
if there exists a partial recursive function ${f :  \mathbb{N} \times I
\rightarrow P_f(I)}$, where $P_f(I)$ is (a recursive encoding of) the
finite subsets of $I$,  so that:
\[ x \in A \iff \exists n, f(n,x)\text{ is defined and }f(n,x) \subseteq B\]
	
When we want to emphasize the function $f$, we will say that 
 $A$ is enumeration-reducible via $f$, written $A \leq_e^f B$.
 (Note that $A$ is uniquely defined given $f$ and $B$).
\end{definition}
To simplify notations, we will write ``$f(n,x) \subseteq B$'' as a shortcut for ``$f(n,x)$ is defined and  $f(n,x) \subseteq B$.''

The intuition is that, for each $n$, $f(n,x)$ is a (finite) set that witnesses
that $x \in A$. If $f(n, x)$ is defined for no value of $n$, there is no possible witness that $x \in A$.

Intuitively, if we are given the elements of $B$ one at a time, we can determine if $x \in A$ by
enumerating all (possibly infinitely many) possible sets of witness $f(n,x)$ and see if one of them contains only elements that
we know are in $B$. If it is not the case, we wait until we get more elements of $B$, or more
possible sets of witnesses.

In particular if $A \leq_e B$ with $B$ recursively enumerable, then
$A$ is recursively enumerable.

Notice that, at any instant in an enumeration of $B$, we obtain some information that some elements
are in $B$, but no information on which elements are not  in $B$. Somehow, $A \leq_e B$ means that we are
able to enumerate $A$ given only \emph{positive information} on $B$.
In particular the complement $\overline{B}$ of $B$ is usually not enumeration reducible to $B$.
Similarly $A \leq_e B$ does not imply that $\overline{A} \leq_e \overline{B}$.

\begin{proposition}
	\label{prop:rec}
        Let $V$ be a quasivariety and ${\cal C}$ its corresponding closure operator.

Then  ${\cal C}(Y) \leq_e Y$ for any $Y \subseteq I$.        
More precisely there exists a partial recursive function $f$ (depending only
on $V$) so that ${\cal C}(Y) \leq_e^f Y$ for any $Y \subseteq I$.

In particular, a point $X$ is enumeration-reducible to any of its presentations.
In particular a finitely/recursively presented point is recursively
enumerable (as a subset of $I$).
\end{proposition}
From the point of view of logic system, this means we can enumerate the set of all consequences of
$Y$ from an enumeration of $Y$.

In particular it means that the set $X$ itself is the smallest presentation of the point $X$ in
terms of information content, as any other presentation can compute this particular presentation.

\begin{example}
The set of valid formulas in finitely and recursively axiomatisable
theories is recursively enumerable.

Subshifts of finite type and effectively closed subshifts have a
recursively enumerable set of forbidden words.
\end{example}
\begin{proof}
Let $V$ be defined by a set $S$, and $S'$ be the set of all statements of the form
\[ a \in X \wedge b \in X \wedge .. \wedge c \in X \rightarrow z \in X\]
that are logical consequences of (finitely many) statements of~$S$.

As $S$ is recursively enumerable, it is easy to see that $S'$ is also recursively enumerable, and of
course define the same quasivariety\footnote{It is easy to see that (if we drop the hypothesis for the
  base set $S$ to be recursively enumerated) the set of all 
  quasivarieties $V$ over a set $I$ may be given itself  the structure of a
  quasivariety $\cal V$, where each quasivariety $V \in {\cal V}$ is identified
  with the set $S'(V)$ of all implications true in any point of the
  quasivariety. The map $S \rightarrow S'$ is of course the closure operator in this new quasivariety.}.
Note that $S'$ contains in particular all implications of the form ``$a \in X \rightarrow a \in X$''.

 On input $a \in I$, consider all finite sets $F$ s.t. the statement  ``$\bigwedge_{i \in F} i \in X \rightarrow a \in X$''
 is a statement of $S'$.
 As $S'$ is recursively enumerable, we can recursively enumerate all such statements, i.e.  there is
 a partial recursive function ${f: I \times \mathbb{N} \rightarrow P_f(I)}$ such that $f(a,\mathbb{N})$ is
 exactly the set of all possible $F$.

Then it is clear that $A \leq_e^f B$ iff $A = {\cal C}(B)$.
\end{proof}
Therefore there is a procedure $f$ that, given any enumeration of any presentation
of $X \in V$, gives an enumeration of $X$. This means that the closure operator, while not
computable in a traditional sense, is computable as an enumeration operator.

With this theorem, we get a new definition of a quasivariety
\begin{theorem}[Alternate Definition 3]
  Let $I$ be a recursive set and $(I, {\cal C})$ is a Tarski space.

    Then the image of ${\cal C}$ is a quasivariety iff there exists $f$ s.t.
    ${\cal C}(Y) \leq_e^f Y$ for all $Y \subseteq I$.
\end{theorem}  
So quasivariety are closure spaces where the closure operator is computable as an enumeration operator.

\begin{proof}
  By definition, recall that $A \leq_e^f B$  means that ${x \in A \iff \exists n, f(n,x) \subseteq B}$.

Consider the set $S$ of all Horn formulas 
  \[  \bigwedge_{j \in f(n,i)} j \in X \rightarrow i \in X  \]
for all $n,i$ whenever $f(n,i)$ is defined.
$S$ is clearly recursively enumerable
  and defines a quasivariety $V$.
  It remains to show that $V$ coincides with the image of ${\cal C}$.

  Suppose that $X$ is in the image of ${\cal C}$, i.e. $X = {\cal C}(X)$ and therefore
  $i \in X  \iff \exists n, f(n,i) \subseteq X$. Keeping only one direction of this equivalence, we
  get that for all $n$ and all $i$, $[f(n,i) \subseteq X  \Rightarrow i \in X]$. Therefore  $X \in V$ by definition of the   quasivariety $V$.

  Conversely suppose that $X \in V$ and let $Y = {\cal C}(X)$.
  As $Y \leq_e^f X$ we know that $i \in Y  \iff \exists n, f(n,i) \subseteq X$. 
  If $i \in Y$, then there exists $n$ s.t. $f(n,i) \subseteq X$ (in particular $f(n,i)$ is defined) which implies by
  definition of $V$ that $i \in X$.
  We conclude $Y \subseteq X$, i.e. ${\cal C}(X) \subseteq X$ and therefore ${\cal C}(X) = X$ and
  $X$ is in the image of ${\cal C}$.    
\end{proof}

\subsection{Craig's theorem}

Craig's theorem is a classical theorem that states that recursively
axiomatisable theories admit recursive presentations.  This theorem is
also valid for subshifts and groups, and can be obtained in the
following way:

\begin{definition}
  A quasivariety $V$ has redundant axioms if for any finite set $S$,
  there exist a set $T$ disjoint from $S$ s.t. ${\cal C}(S) = {\cal S}(T)$
\end{definition}
The definition implies that no finite set of axioms is necessary, as
it can always be replaced by another equivalent set of axioms.
This second set can be taken finite:
\begin{proposition}
  Suppose that ${\cal C}(S) = {\cal S}(T)$ for some finite set $S$. Then
  
there exists a finite subset $T'$  of $T$ s.t. ${\cal C}(S) = {\cal C}(T)$
\end{proposition}
\begin{proof}
  Let $ S = \{ s_1 \dots s_n\}$.
As the closure operator is finitary, every element $s_i$ of $S$ is in the
closure of a finite subset $T_i$  of $T$.

We now take $T' = \cup T_i$.

\end{proof}

\begin{example}
In the following examples, we will give, for any axiom $p \in I$, a
infinite collection of disjoint finite sets, each of them being equivalent
ot $p$. This obviously proves the result. 
  \begin{itemize}
    \item The quasivariety of theories has redundant axioms. Indeed we
      can always replace an axiom $\phi$ by any axiom of the form $\overbrace{\phi \wedge
        \phi \wedge \dots \phi}^n $
      
    \item The quasivariety of groups has redundant axioms. Indeed we
      can always replace an axiom $u$ by any axiom of the form
      $vuv^{-1}$ for some $v$ 
    \item The quasivariety of subshifts has redundant axioms. Indeed we
      can always replace an axiom $u$ by any set of axioms $\{ uv, v
      \in \Sigma^n\}$ for some $n$.
      
    \item The quasivariety of closed subsets has redundant      
      axioms. Indeed we      
      can always replace an axiom $]p,q[$ by the two axioms
      $]p,r[,]s,q[$ for any $r \leq q$ and $s \geq p$      
      
    \end{itemize}   
\end{example}

\begin{proposition}
  Suppose that $V$ is a quasivariety with redundant axioms. There
    exists a total recursive  function $g$ that given an axiom $p \in I$
  and a finite set $S$ outputs a set $T$ s.t. $C(T) = C(p)$ and $T$
  is disjoint from $S$.
  
\end{proposition}
\begin{proof}
  
  Such a $T$ exists by hypothesis.  
  Notice that $EQ = \{ (T,T') |  C(T) = C(T') \}$ is recursively
  
  enumerable.  
  Therefore the recursive function $g$ on input $p$ and $S$ only has to look at all possibles sets $T$ that are
  disjoint from $S$ until it finds one that is equivalent to $p$.  
\end{proof}

\begin{theorem}[Craig's theorem for quasivariety with redundant
  axioms]
  
  Suppose that $V$ is a quasivariety with redundant axioms. If $X$ is
  recursively presented ($X = C(R)$ for some recursively enumerable
  set $R$), then there exists a recursive set $R'$ s.t. $X = C(R')$.
\end{theorem}
\begin{proof}
Without loss of generality, we identify the core set $I$  with the set
of integers $\mathbb{N}$.

Suppose $R\not=\emptyset$ is the image of a total recursive function
$R = \{ f(n), n \in \mathbb{N}\}$.
We inductively build a sequence of finite sets $T_i$  by
$T_{i+1}  =  g(f(i+1), [0,\max T_i])$.
That is, $T_{i+1}$ is a finite set equivalent to $f(i+1)$ that uses
only integers greater than any seen previously.

By construction $\bigcup T_i$ is a presentation of $X$.
Furthermore, $\bigcup T_i$ is not only recursively enumerable, but
recursively enumerable in an increasing order, which makes it recursive.
\end{proof}

The trivial quasivariety (given by no Horn formulas) show this theorem
is not true in general without any assumptions.

\section{Maximal elements}

We will show how some structural properties on some particular
points of a quasivariety translate into computability properties.
We start with maximality.
\begin{definition}
	An element $X$ of a quasivariety $V$ is maximal if $X \subseteq	Y$, 
	with $Y \in V$, implies $Y = X$ or $Y = I$.
\end{definition}
The case $X = I$ is a degenerate case that is usually not interesting
in the applications.

\begin{example}
  In the quasivariety $V$ of theories $X \subseteq Y$ means that $X$
  has less theorems than $Y$, i.e that $Y$ is an extension of $X$.  
  Therefore maximal points in the quasivariety $V$ of  theories are exactly the \emph{complete} theories (plus the inconsistent theory).

  In the quasivariety $V$ of subshifts, if $X$ and $Y$ are
  respectively the words
  that do not appear in   subshifts $S$ and $T$, then $X \subset Y$
  means that $T \subseteq S$. Therefore maximal points in the quasivariety $V$ of  subshifts are called
  \emph{minimal} subshifts   (plus the empty subshift).

  In the quasivariety $V$ of f.g. groups with $n$ elements, if $X$ and $Y$ are respectively the normal
  subgroups that code $F$ and $G$ then $X \subseteq Y$ means that
  $G$ is a quotient of $F$.  
  Therefore maximal points in this quasivariety are exactly the simple
  groups (plus the trivial group).

  The case of the quasivariety $V$ of closed subsets of
  $\mathbb{R}$ is similar to the case of subshifts, and we get that 
  maximal points in this quasivariety are exactly the singletons $\{x\}$ (plus the empty set).    
\end{example}

Our first theorem generalizes the theorem of Kuznetsov, without
any hypothesis on the computability of the structure:

\begin{theorem}
\label{thm:max}	
	Let $V$ be a quasivariety.	If $X \in V$ is maximal then $\overline{X} \leq_e X$.
\end{theorem}
Recall that $\overline{X}$ is the complement (in $I$) of $X$.
A set $X$ s.t. $\overline{X} \leq_e X$ is sometimes called a \emph{total} set.

Before going to the (easy) proof, let us explain the significance of
the theorem, and why it generalizes previous theorems.
$\overline{X} \leq_e X$ means that, from any enumeration of $X$, one
can compute some enumeration of $\overline{X}$: Given positive
information on $X$ (which elements are in $X$) we can get  negative
information on $X$ (which elements are not in $X$).

This does not happen quite often. For an example,
suppose that $X$ is recursively enumerable: there is an algorithm that
produces an enumeration of the elements of $X$. If $\overline{X} \leq_e X$, we get that $\overline{X}$ is also recursively enumerable,
and therefore $X$ is recursive.
So the only recursively enumerable sets $X$ for which
$\overline{X}\leq_e X$ are the recursive sets.
In fact, using a well known theorem of Selman\cite{Selman},
$\overline{X}\leq_e X$ can be reformulated as ``For every oracle $A$,
if $X$ is recursively enumerable in $A$, then  $X$ is recursive in $A$''.

Notice that existing theorems \cite{Kuznetsov,Hochman2, BooneHigman}, when rephrased in our vocabulary
are usually of the form ``Assume that $X$ is finitely presented. Then
if $X$ is maximal, it is computable''. Our theorem is more general, as
we do not have any assumption about the presentation of $X$ (in
particular $X$ might not be recursively enumerable).

We will now give a brief idea of the proof using the language of
logic.
Suppose that $X$ is a maximal set of consistent formulas.
By definition, $\phi \not\in X$ if adding $\phi$ to $X$ is
inconsistent.  This is equivalent to saying that there exists a finite
$Y \subseteq X$ s.t. $\{ \phi\} \cup Y$ is inconsistent.
To know whether $\phi \not\in X$, it is therefore sufficient to list
all possible $Y$ s.t. $\{ \phi\} \cup Y$ is inconsistent, and see if
one of them is included in $X$.
The proof mimics closely this idea.

\begin{proof}
	If $X = I$ the result is obvious by taking $f(n,x)$ to be nowhere defined.
	Otherwise, let $a \not\in X$.
	
	Then $x \in\overline{X}$ iff the smallest point containing both
	$X$ and $x$ contains $a$.
	
	Recall there is a function $f$ so that  $A \leq_e^f B$ iff $B$ is a presentation of $A$.

    Thus $x \in \overline{X}$ iff there exists $n$ so that $f(n,a) \subseteq X \cup \{x\}$.

Let $g$ be the partial recursive function defined by  ${g(n,x) = f(n,a) \setminus \{x\}}$ whenever $f(n,a)$ is defined.
	
	Then $x \in \overline{X}$ iff there exists $n$ so that $g(n,x) \subseteq X$.
\end{proof}	

\begin{corollary}
  Let $V$ be a quasivariety. If $X$ is finitely (or recursively)
presented and maximal, then $X$ is recursive.
\end{corollary}

\begin{example}
Let $S$ be a complete theory. Then the set of formulas that are invalid
in $S$ is enumeration-reducible to the set of formulas
that are valid in $S$.
In particular, if $S$ is finitely axiomatisable (or recursively
axiomatisable), then the set of formulas that are valid in $S$ is
recursive \cite{Janiczak}.

Let $S$ be a minimal subshift. Then the set of words that appear in
$S$ is enumeration-reducible to the set of words that do not appear in
$S$. In particular if $S$ is a minimal
subshift of finite type (or a minimal effectively closed subshift),
the set of forbidden words of  $S$ is recursive.
This theorem was first proven in \cite{Hochman2,ballier2008}.

Let $G$ be a f.g. simple groups . Then the complement of the word problem of $G$ is
 is enumeration-reducible to the word problem of $G$. In particular if
 $G$ is a finitely presented simple group, the  word problem of $G$ is
 recursive.
 This theorem was first proven in \cite{BooneHigman}.

 Let $x \in \mathbb{R}$. Then the set of open rational intervals $]p,q[$ that
  contains $x$ can be enumerated from the set of open rational intervals that
  do not contain $x$. In particular, if $x$ is recursively
  presented, i.e. if $\{x\}$ is a $\Pi_1^0$ class, then $x$ is
  recursive. See for example \cite{CClo}.
\end{example}

The astute reader may realise that in the case of first order logic there is
actually an easier proof of the theorem: Indeed $\phi \not\in X$ iff $\neg \phi \in X$.
This means we have a stronger reduction: $\overline{X}$ is many-one reducible to
$X$, and in this particular case, the reverse is also true: $X$ is many-one
reducible to $\overline{X}$.  This does not hold in general. To see why, let's look at a variant of
the quasivariety $V_{FO}$. If we look only at $\forall\exists$ formulas, we get the following:
If $T$ is a complete theory which is
$\forall \exists $-axiomatizable,
then we can enumerate the $\forall \exists$ formulas that are false from any enumeration of the $\forall \exists$ formulas
that are true. This cannot be proven by using the $\neg$ operator as a magic
wand, as the negation of a $\forall \exists$ formula is not a $\forall \exists$
formula.

In fact our theorem is best possible: If $X$ is a
maximal point in a quasivariety, the fact that $\overline{X}\leq_e X$
is the strongest statement we can prove in full generality on $X$.
Indeed there is a converse: If a set $X$ satisfies $\overline{X} \leq_e X$, then it is a maximal point in a
suitable quasivariety:
\begin{theorem}
  \label{thm:conv}
  Let $A \subseteq I$ with $I$ recursive.
  
	Let $\overline{A} \leq_e A$. Then there exists a quasivariety $V$
	s.t. $A \in V$ and $A$ is maximal.
\end{theorem}
Note that recent (unpublished) results of Ethan McCarthy show that we can take $V$ to be the
quasivariety $V_{sym}$ of subshifts, in the sense that any set $\overline{A}$
s.t. $\overline{A} \leq_e A$ is enumeration-equivalent to (the
language of) a minimal subshift.

\begin{proof}
	Let $\overline{A} \leq_e^f A$.	
	Thus $x \in \overline{A}$ iff there exists $n$ s.t. $f(n,x) \subseteq A$.
	
	Let $V$ be the quasivariety defined by all axioms
        \[
        x \in X \wedge \bigwedge_{i \in f(n,x)}
		  i \in X \rightarrow j \in X.
                  \]
                  for all $n \in \mathbb{N}, x,j \in I$ whenever  $f(n,x)$ is defined.
	
	$A$ is in this quasivariety: Indeed, there is no $x \in A$ and $n$ so
	that $f(n,x) \subseteq A$, thus all premises are false.
	
	It is clearly maximal: Let $A \subseteq A'$ and $a \in A' \setminus A$.
	Then $a \in \overline{A}$ thus there exists $n$ s.t. $f(n,x)
	\subseteq A \subseteq A'$ thus for all $y$, $y \in A'$, thus $A' = I$.	
\end{proof}

Notice that the proof of Theorem~\ref{thm:max} was nonuniform: We need to exhibit, for $X$ maximal and $X \not= I$,
some element $a$ that is not in $X$, and this $a$
might depend on $X$.
In many examples, $a$ can be chosen independently of $X$.
\begin{theorem}[Uniform version]
Let $V$ be a quasivariety so that $I$ (the whole set) is finitely presented.
Then there exists a partial recursive function $g$ so that if $X$ is maximal, $X \not= I$, then $\overline{X} \leq_e^g X$.
\end{theorem}

\begin{proof}
  Let $f$ be the partial recursive function so that
  $A \leq_e^f B$ iff $B$ is a presentation of $A$.
  
  By definition, there exists a finite set $E$ s.t. any point containing all of $E$ is
equal to $I$. Say $E = \{ e_1 \dots e_k\}$.

Thus $x \in\overline{X}$ iff for all $e \in E$, there exists $n$ so that $f(n,s) \subseteq X \cup\{x\}$.

Let $g(n_1, n_2 \dots n_k,x) = \cup_{i \leq k} f(n_i,e_i) \setminus \{x\}$.

Then $x \in \overline{X}$ iff there exists $n_1 \dots n_k$ so that $g(n_1 , n_2 \dots n_k, x) \subseteq X$.

\end{proof}	

\begin{example}
Let $V$ be the quasivariety of first order theories. $\{ \exists
x , x \not= x\}$ is a finite presentation of the inconsistent theory,
i.e. I itself. 

Let $V$ be the quasivariety of subshifts. $\{\epsilon\}$ is a finite
presentation of the empty subshift, i.e. I itself: If 
we forbid the empty word $\epsilon$ to appear, we have essentially forbidden
all words to appear.

Let $V$ be the quasivariety of groups with $n$ generators $g_1, g_2,
\dots g_n$. Then $\{g_1, g_2, \dots g_n\}$ is a finite
presentation of the trivial group, i.e. I itself.

As we saw earlier, the closed subset $\emptyset$ is not finitely
presented in the quasivariety $V$ of closed subsets of $\mathbb{R}$,
as any finitely presented closed set contains neighborhoods of $\pm
\infty$.

This means that, in the first three cases, the reduction is
uniform. In particular there exists an algorithm that, given a minimal
subshift of finite type $S$ (resp. a finitely axiomatisable complete
theory $S$, a f.g. simple group) computes the set of forbidden words
of $S$ (resp. the set of valid formulas in $S$, the word problem of
$G$).

The reduction is not uniform in the last case: To enumerate the set of
open rational intervals $]p,q[$ that contain $x$ from the set of open
rational intervals that do not contain $x$, we need an additional
information about $x$, e.g. in the form of a finite interval that contains
it.  This phenomenon is well known in computable analysis.

\end{example}

We present here a slight generalization of the main theorem. Instead of
requiring that $X$ is maximal, we require that we have  an exact description of
all points above $X$:
\begin{theorem}
	\label{thm:below}
	Let $(S_n)_{n \in \mathbb{N}}$ be a recursively enumerable collection of
	finite subsets of $I$, and $(Y_n)_{n \in \mathbb{N}}$ be the points
	presented by $S_n$.
	Let ${\cal Y} = \{ Y_n, n \in \mathbb{N}\}$.
	
	We say that $X$ is maximal below $\cal Y$ if 
 $X \not\in \cal Y$, but every point larger than $X$ is in $\cal Y$.

	Then $\overline{X} \leq_e X$, uniformly.	
\end{theorem}
Notice that it is not required that all elements of ${\cal Y}$ are larger than
$X$.
In the examples one is usually interested in a set ${\cal Y}$ of ``well-known''
points, and we look at all points that are above this set of well-known points.

In the theory of algebra, one would take for ${\cal Y}$ all finite algebras.
Then a point $X$ below $\cal Y$ is a point $X$ that is not a finite algebra but for which
every quotient of $X$ is finite.
A specific version of this theorem in this context was proven by
Maltsev \cite{Maltsev}: Every finitely generated and finitely presented algebra on which all congruences have finite index is recursive.

\begin{example}
Let's call a subshift just-infinite if it is infinite but all its proper
subshifts are finite.
For example, the set of infinite words over the alphabet $\{0,1\}$ with at most one
 occurence of the symbol $1$ is a just-infinite subshift.
It is easy to see that a just-infinite subshift satisfy the property above.
In particular, for a just-infinite subshift of finite type (or any
effectively closed just-infinite subshift), the set of forbidden
patterns is recursive.

A group $G$ with generators
$a_1,\dots a_k$ is finite iff there exists a size $p$ s.t. all words
of length $p$ over the alphabet $\{a_1^{\pm 1} \dots a_k^{\pm 1}\}$
are equal to a word of smaller length.
As a consequence, the set of all finite presentations of finite groups is
recursively enumerable.
A group for which all proper quotient are finite is called a
\emph{just-infinite} group. Said otherwise, a just-infinite group is
maximal below the finite groups.
We get from the previous theorem that if $G$ is a just-infinite recursively presented group, then the word
problem for $G$ is recursive.

\end{example}	
\begin{proof}   
	Straightforward generalization of the previous theorem.
	
	Let $(S_n)_n$ be the recursive collection of finite sets, and
	write $S_n = \{ a^n_1, \dots, a^n_{h(n)}\}$, where $h(n)$ is the
	(computable) size of $S_n$.

	Recall there is a function $f$ so that  $A \leq_e^f B$ iff $B$ is a presentation of $A$.
	
	Then $x \in \overline{X}$ iff the smallest point (or any point)
	containing  $\overline{X} \cup \{x\}$ is one of the sets of $\cal Y$.
		
	Thus $x \in \overline{X}$ iff $\exists p, \exists n_1 \dots
	n_{h(p)}, \cup_{i \leq h(p)} f(n_i, a^p_i) \setminus \{x\} 	\subseteq X$.
\end{proof}

\section{Discriminable points}
Discrimination is a generalization of maximal elements.
The concept and the vocabulary comes from group theory, in particular
\cite{Cornulier}.
The notion is already present in Kuznetsov\cite{Kuznetsov}, where the
author defines a concept of a completely finitely presented algebra,
which corresponds in our vocabulary to a point which is both finitely
presented and finitely discriminated.

\begin{definition}
Let $V$ be a quasivariety.

	A set $Y$ is a discriminator for a point $X \in V$ if $Y \cap X =	\emptyset$
	and for every point $X' \in V$ s.t. $X \subsetneq X'$, we have $X' \cap Y \not= \emptyset$.
			
	$X$ is \emph{finitely discriminable} if it admits a finite discriminator.
	$X$ is \emph{recursively discriminable} if it admits a recursively enumerable discriminator.
\end{definition}
So a discriminator $Y$ is a set of objects that are not in $X$ but s.t. every
extension of $X$ contains at least some element of $Y$.

Notice that by definition $\overline{X}$ is always a discriminator for $X$.
If $X$ is maximal, any \emph{single} element in $\overline{X}$ acts as a
discriminator for $X$. Said otherwise, every (nontrivial) maximal element is
finitely discriminable.

When $V$ is seen as a topological space, points $X$ that are both finitely discriminable and finitely presented 
are isolated: If $R$ is the presentation and $Y$ the discriminator, then $X$ is
the only point of $V$ that contains $R$ and does not intersect $Y$. Isolated points have been particularly studied in the theory of
groups \cite{Cornulier}.

Easy examples come from the following proposition:
\begin{definition}
  Let $V$ be a quasivariety.
  A point $X \in V$ is \emph{quasi-maximal} if there are only finitely many points in $V$
  above $X$. 
\end{definition}
\begin{proposition}
  Quasi-maximal points are finitely discriminable.
\end{proposition}
\begin{proof}
  Choose for every point $X' \supsetneq X$ some $a_{X'} \in X' \setminus X$ and
  take $Y = \{ a_{X'}, X' \supsetneq X\}$.
\end{proof}  
  
There are examples of finitely discriminable points that are not
quasi-maximal. See  \cite{PISFT} for an example in the quasivariety
of two-dimensional subshifts, or  \cite[Theorem 5.3]{Cornulier} for an example
in the quasivariety of groups.

On the other hand:
\begin{proposition}
In the quasivariety of theories, finitely discriminable points are exactly the
quasi-maximal points.
\end{proposition}
\begin{proof}
  Let $T$ be a finitely discriminable theory. This means there exists $Y = \{
  \phi_1 \dots  \phi_n \}$
  s.t. no $\phi_i$ is in $T$ and every extension of $T$ contains
  some  $\phi_i$.

  First, note that if $T \models \phi_i \rightarrow \phi_j$ for some $i \not=j$ then every extension
  that contains $\phi_i$ also contains $\phi_j$  so that $Y \setminus \{
  \phi_i\}$ is also a discriminator for $T$.
  We can therefore suppose wlog that if $\phi_i$ and $\phi_j$ are in $Y$, then
  $\phi_i \rightarrow \phi_j$ is not in $T$ unless $i = j$.

  We want to prove that $T$ has only finitely many extensions.
  It is sufficent to prove that  $T$ has only finitely many
  \emph{complete} extensions.
  Indeed, a theory is  entirely characterized by its set of complete
  extensions: If $T_1$ and $T_2$ are two different consistent theories there
  exists $\varphi$ s.t. $\varphi \in T_1$ and $\varphi\not\in T_2$
  (or conversely). As $T_2$ does not prove $\varphi$, the theory $T_2
  \cup \{\neg\varphi\}$ is not inconsistent and therefore has a
  complete extension. This complete extension cannot be an extension of $T_1$.

  First we examine the formula $\vartheta = \neg \phi_1 \vee \neg \phi_2 \vee \dots \neg \phi_n$.
  Suppose that $\vartheta \not\in T$. Then by discriminability there  exists $i$ s.t.
  $T \cup \{\vartheta\} \models \phi_i$.
  But then $T \models \neg \phi_i \rightarrow \phi_i$ and therefore $T \models
  \phi_i$, which is impossible by discriminability.
  Therefore $\vartheta \in T$, that is $T \models \neg \phi_1 \vee
  \neg \phi_2 \vee \dots \vee
  \neg \phi_n$.

  Next, fix some formula  $\phi_i \in Y$, and let $\psi$ be any formula.
  We look at $\phi_i \vee \psi$. There are two cases:
  \begin{itemize}
    \item $\phi_i \vee \psi \in T$. Therefore $T \cup \{ \neg \phi_i\} \models
      \psi$.
    \item Otherwise by discriminability there exists $j$ s.t. $T \cup \{ \phi_i
      \vee \psi\} \models \phi_j$.
      This means that $T \cup \{ \phi_i\} \models \phi_j$ and by our
      supposition this implies $\phi_i = \phi_j$.
      Therefore $T \cup \{ \phi_i \vee \psi \} \models \phi_i$. In particular $T \models \psi \rightarrow \phi_i$ and therefore $T \cup \{ \neg
      \phi_i\} \models \neg \psi$            
  \end{itemize}
  We have therefore proven that for any formula $\psi$, either $T \cup  \{\neg\phi_i\} \models \psi$ or $T \cup \{\neg\phi_i\} \models \neg\psi$.
  Therefore $T \cup \{\neg\phi_i\}$ is an axiomatisation of a complete theory (or is inconsistent).
  
  This proves the result: if $T'$ is a complete extension of $T$, then by the
  first point, some $\neg\phi_i$ should be true in $T'$. As $T \cup
  \{\neg\phi_i\}$ is an axiomatisation of a complete theory, this means that
  $T'$ is actually the closure of $T \cup \{\neg\phi_i\}$. Therefore there are at
  most $n$ complete extensions of $T$.    
\end{proof}  

We now go to the generalization of the theorem to recursively discriminable points:
\begin{theorem}
  \label{thm:disc}
  Let $V$ be a quasivariety and $X$ a point in $V$.
	If $X$ is recursively discriminable, then $\overline{X} \leq_e X$.
\end{theorem}
\begin{corollary}
  If $X$ is quasi-maximal and recursively presented, then $X$ is recursive.
\end{corollary}
This corollary was first obtained for subshifts in \cite{Salo4}.
\begin{proof}
Let $Y$ be the discriminator.
Now $x \in \overline{X}$ iff the point presented by $X \cup\{x\}$
contains some element of $Y$.

Recall again that there exists a partial recursive function $f$ s.t. $A
\leq_e^f B$ iff $B$ is a presentation of $A$.
Thus $x \in \overline{X} \iff \exists y \in Y, \exists n \in N, f(n,y) \subseteq X \cup \{x\}$.

Let $Y = (y_m)_{m \in \mathbb{N}}$ a recursive enumeration of $Y$.

We now define $g(n,m,x) = f(n,y_m) \setminus \{x\}$.
Then $x \in \overline{X} \iff \exists n,m, g(n,m,x) \subseteq X$.
\end{proof}

\begin{corollary}
	Let $V$ be a quasivariety and $X$ a point of $V$.
	
	$X$ is recursive iff it is recursively presented and recursively
	discriminable.
\end{corollary}	
This corollary was first obtained for groups in \cite{Cornulier}.

\begin{proof}
	If $X$ is recursive, then $X$ is a presentation of $X$ which is
	recursive, and $\overline{X}$ is a discriminator for $X$ which is
	recursive.
	
	Conversely, if $X$ is recursively presented by $Y$, then $X$ is
	recursively enumerable, as $X \leq_e Y$.
	As $X$ is recursively discriminated, $\overline{X} \leq_e X$, thus
	$\overline{X}$ is recursively enumerable, and $X$ is recursive.
\end{proof}

\section{Difference between $X$ and $\overline{X}$}

All results in the previous sections show that all consequences of maximality or
quasi-maximality are of the same form: from any enumeration of $X$ we can
compute some enumeration of $\overline{X}$.
As the converse is usually not true, this means that in all these examples we
get strictly more information from $X$ than from $\overline{X}$.

In this last section, we will try to explain what is this information.
For this we need to ask a bit more from the reduction:
\begin{proposition}
Let $X \in V, X \not=I$ s.t. $\overline{X} \leq_e X$. Then there exists a total recursive function $f$  s.t. $\overline{X} \leq_e^f X$.  
\end{proposition}
(That is, $f$ can be chosen total).
\begin{proof}
  $\overline{X}\leq_e^g X$ for some partial recursive function $g$.
  As $g$ is partial recursive, there exists a total recursive function $h: \mathbb{N} \times \mathbb{N} \times I \rightarrow \{0,1\}$ s.t $g(n,x)$ is defined iff there exists $m$ s.t. $h(n,m,x) = 1$.
  
  Let $a \not\in X$. Define $f(n,m,x)$ by $f(n,m,x) = \{a\}$ if $h(n,m,x) = 0$ and
  $f(n,m,x) = g(n,x)$ otherwise.
  $f$ is total  and
  $  x \in \overline{X} \iff \exists n,m\  f(n,m,x) \subseteq X$.
\end{proof}
Notice that if $I$ is finitely presented by a set $A$, we can replace
$\{a\}$ by $A$ and obtain a uniform version of the theorem, i.e. $f$
does not depend on $X$.

\begin{definition}
	Let $\overline{X} \leq_e^f X$ with $f$ total.
	Let $g(x) = \min \{n \in \mathbb{N} | f(n,x) \subseteq X\}$.
	
	$g$ is a partial map, as it is defined only on $\overline{X}$.
	We identify $g$ with $G = \{ (x,g(x)) | g(x) 	  \text{is defined}\}$
\end{definition}	

As $f(n,x)$ is the list of all possible witness that $x \in \overline{X}$,
$g(x)$ therefore represents the very first real witness  in the enumeration.
Note that $g$ depends of course on the enumeration $f$.
We will now see that a bound on $g$  is exactly the information we need to recover $X$ from $\overline{X}$.

\begin{proposition}
	Let $\overline{X} \leq_e^f X$ with $f$ total.
Then $G \leq_e X$.
\end{proposition}	
\begin{proof}
	
Suppose we are given an enumeration of $X$, and at some point we conclude
that $x \in \overline{X}$ because $f(n,x) \subseteq X$ for some $n$,
and  all elements of $f(n,x)$ are currently known to be in $X$.
Then we know that $g(x) \leq n$.
We then look at all sets $f(i,x)$ for $i < n$.
At some point in our enumeration, we will know the status of all points
in $\cup_{i < n} f(i,x)$, either because they were enumerated in $X$,
or we were able to prove that they are in $\overline{X}$.
Thus we will be able to determine the exact value of $g(x)$.
\end{proof}

\begin{proposition}
	Let $h$ be a total function that dominates $g$:	$h(x) \geq g(x)$ whenever $g(x)$ is defined.
	We identify $h$ with the total set $H = \{ (x,h(x)), x \in I \}$
	
	Then $X \leq_e H \oplus \overline{X}$
\end{proposition}
\begin{proof}
	$x \in X \iff \forall n \leq h(x), f(n,x) \cap \overline{X} \not=
	\emptyset$.	
\end{proof}	
It is important to note that it is not $g$ itself which is
important, but any upper bound on $g$.
Note also that the reduction in the theorem is stronger than
enumeration reducibility.

The previous propositions may seem uninteresting, as the exact
definition of $g$ depends on the particular operator $f$ that was used
to prove that $\overline{X} \leq_e X$, and there does not seem to be
any canonical way to associate some map $g$ to every set $X$ that
satisfies that $\overline{X} \leq_e X$.

However, in many cases, we can say more.
In particular, in the interesting case of a quasivariety $V$ where $I$ is
finitely presented, $f$ does not depend on $X$ but only on $V$, so that
$g$ can be defined indeed in a canonical way

\begin{example}
A minimal subshift $S$ has a \emph{quasiperiodicity}
function (also called uniform recurrence function): 
There exists a function $g(n)$ s.t. every word $w$ of size $n$ that
appear in $S$ is contained in every word of size $g(n)$ that appear in
$S$.
For minimal subshifts, the previous propositions may thus be interpreted this way:
\begin{itemize}
	\item The set of words that appear can be obtained from an
	  enumeration of the set of words that do not appear
	\item The quasiperiodicity function can be obtained from an
	  enumeration of the set   of words that do not appear
	\item The set of words that do not appear can be obtained from an
	  enumeration of both the quasiperiodicity function and the set of words that appear.
\end{itemize}
In particular, a minimal subshift of finite type has a recursive
quasiperiodicity function. This theorem was first established in \cite{ballier2010}.
This theorem is optimal in the sense that it is easy to find minimal subshifts
which have a recursive quasiperiodicity function but which are not
recursive, and minimal subshifts for which the set of words that
appear is recursively enumerable but not recursive.
\end{example}

We now introduce a notion which seems new in group theory for simple groups:

\begin{definition}
Let $S$ be a finitely generated, simple group, with generators $a_1 \dots a_k$.

For a word $w$ over the generators $a_1 \dots a_k$, let $g(w)$
be the smallest $p$ s.t. all generators $a_i$ can be written as products
of less than $p$ elements of the form  $hwh^{-1}$ or $hw^{-1}h^{-1}$ for $h
\in S$, and each  $h$ is a product of less than $p$ generators.
$g(w)$ is defined only when $w$ is not the identity element on $S$.

We now define $g(n) = \max \{g(w) | w \in B_n, w\not=1\}$, where
$B_n$ is the set of elements of $G$ that can be written as a product
of less than $n$ generators.
\end{definition}
$g(w)$ is well defined if $w\not=1$.
Indeed the set of all elements that can be written as above
is a normal subgroup of $S$ that is nontrivial (it contains $w$), 
and thus is equal to $S$.

Note that $g$ depends on the choices of generators of $S$, but it is
easy to see that different choices of $G$ only changes the function
upto a linear factor.

Then, in $S$:
\begin{itemize}
	\item The complement of the word problem on $S$ can be enumerated
	  from an enumeration of the word problem on $S$.
	\item The function $g$ can be computed from an enumeration of
	  the word problem on $S$.
	\item The word problem on $S$ can be enumerated from both the
	  complement of the word problem on $S$ and any bound $t$ on $g$.
	  Indeed $w = 1$ iff there exists a generator $a$ s.t. all
	  products of less than $t(|w|)$ terms of the form $hwh^{-1}$, where
	  each $h$ is the product of less than $t(|w|)$ generators, are
	  different from $a$.
  \end{itemize}	
This new function $g$ is therefore an equivalent of the quasi-periodicity
function for subshifts.

\section*{Conclusion}

A consequence of this work is the following: many results in algebra
assert that if a finitely presented structure has some property $P$,
then the structure is recursive.
The usual way  these results are done is by proving that
having property $P$ and being finitely presented imply that the
structure is both recursively enumerable and co-recursively enumerable.

However these results can be divided in two parts:
\begin{itemize}
    \item Either they are still valid when the structure is only
	  recursively presented instead of finitely presented. 
	  In which case, as presented here, the result can usually be
	  generalized to obtain a result that hold for any structure with
	  property $P$.
	  From the proof we also obtain that in this case a function $g$
	  can be attached to each structure, that gives additional
	  information about it. This is the case for example for minimal
	  subshifts (where we can attach the \emph{quasiperiodicity}
	  function).
	\item Or they do not generalize to recursively presented
	  structures, which means they really need the structure to be finitely
	  presented to be able to prove that the structure is
	  co-recursively enumerable.
	  In which case it is not clear how these results can be
	  generalized. 
          It is for example the case for residually finite
	  groups (finitely presented residually finite groups have a
          recursive word problem, but there are some recursively presented residually
	  finite groups that are not \cite{Dyson}), or for the analog
	  concept of subshifts whose periodic points are dense
	  \cite{HochMey} (a subshift of finite type whose periodic points
	  are dense has a recursive set of forbidden words, but this is
	  not true for an effectively closed subshift whose periodic
	  points are dense).
\end{itemize}

\section*{Open Questions}

This article presents how a few results in logic and symbolic
dynamics may be related once seen in the concept of universal algebra,
and how they can be generalized for structures that are not
recursively presented.

There exist other theorems which offer a striking similarity, but for
which a general statement is not known, most proeminently 
Higman's embedding theorem and Boone-Higman's theorem. We focus
here the discussion on the former theorem (The author claims he has a
proof of an equivalent of the Boone-Higman theorem for subshifts,
which will be found in a later paper).

\begin{theorem}[\cite{Higman}]
A finitely generated group can be embedded in a finitely presented
groups iff it has a recursively enumerable set of defining relations.
\end{theorem}
\begin{theorem}[\cite{Kleene,CraigVaught}]
An arbitrary theory (with identity) is finitely axiomatisable using
additional predicates iff it is recursively axiomatisable.
\end{theorem}
\begin{theorem}[\cite{Hochman2}, see also \cite{Aubrun2,DRS3}]
A subshift is the subaction of a (projection) of a shift
of finite type iff
it is effectively closed.
\end{theorem}	
The note by \cite{Kuznetsov} also suggests an analog for universal
algebras, and a similar theorem for semigroups also exist.
Note also that the Relative Higman Embedding Theorem
\cite{HigmanScott} also has an equivalent in the domain of subshifts
\cite{AubrunS09}.

Proofs of these theorems are tremendously combinatorial, as each proof needs
to embed a Turing machine (or another computational device) into an
algebraic system, and the methods to do this are quite different.
However the fact remains that all these theorems have similar
hypotheses and conclusions, so that either it is a striking
coincidence, or something deep can be found here.

\section*{Acknowledgements}
The author thanks Dmitry Sokolov for helping him with the
translation of \cite{Kuznetsov}, and Yves de Cornulier for valuable
comments on a first draft of this article.


\appendix
\section{Appendix}
\subsection{Proof of Theorem \ref{thm:pi01}}
\begin{theorem2}[Theorem \ref{thm:pi01}]
	A $\Pi_1^0$ class $S \subseteq \{0,1\}^I$ is a quasivariety iff it
	contains $I$ and is closed under (finite) intersections.
\end{theorem2}	
\begin{proof}
	One direction has already been stated as a fact above.
	
	Now suppose $S$ is a $\Pi_1^0$ class which contains $I$ and is
	closed under (finite) intersection.
	
    Let $\cal F$ be the collection of all partial maps $(f_i)_{i \in
	  \mathbb{N}}$ where $f_i \in \{0,1\}^{F_i}$, with $F_i$ finite, which disagree with every element of $S$.
	
	By definition of a $\Pi_1^0$ class, $\cal F$ is recursively
	enumerable, and every element not in $\cal F$ agrees with at least
	one point of $S$.

	Now let ${\cal F'}$ be the restriction of $\cal F$ to partial maps
	that takes value $0$ in exactly one point. ${\cal F'}$ is also
	recursively enumerable, and the $\Pi_1^0$ class defined by ${\cal
	  F}'$ is by definition a quasivariety $V$. It is clear that $S \subseteq V$, we now prove that
	they are equal. 
	
	For this, suppose $x \in V \setminus S$.   
    Then there exists a map $f \in {\cal F} \setminus {\cal F}'$
	that agrees with $x$.
	
	As the whole set $I$ is in $S$, no partial map taking only the
	value $1$ can be in ${\cal F}$, hence  $f$ must take value $0$ in at least one point.
	
	Let $A$ be the (possibily empty) set of positions where $f$
	takes value $1$, and $B$ the set of positions where $f$ takes
	value $0$.
	As $f \not\in {\cal F}'$, $|B| \geq 2$.

	For each $b \in B$ consider the map $f_b$ defined on $A \cup
	\{b\}$ and taking value $1$ on $A$ and $0$ on $\{b\}$.
	
    Note that $x$ agrees with every map $f_b$, and each such map
	takes value $0$ in exactly one point.   
	As a consequence, none of the map $f_b$ is in ${\cal F}$
	(otherwise it would be in ${\cal F}'$).
	Therefore, for each $b$, there exists a point $y_b \in S$ that
	agrees with $f_b$.
	
	But then $\bigcap_b y_b$ is a point of $S$ that agrees with $f$, a
	contradiction.
\end{proof}

\end{document}